\documentclass[3p]{elsarticle}

\usepackage{lineno,hyperref,color,amsfonts}
\usepackage{amsmath, latexsym, amsthm, siunitx, amssymb, amscd, mathtools, bbold, tikz}
\usetikzlibrary{matrix,decorations.pathreplacing,positioning}

\newcommand{\ints}{\mathbb{Z}}

\newcommand{\re}{\mathbb{R}}
\newcommand{\cx}{\mathbb{C}}

\newcommand{\sub}[1]{\langle #1\rangle} 
\newcommand{\cA}{\mathcal{A}}
\newcommand{\cC}{\mathcal{C}}
\newcommand{\Cl}{\mathcal{C}\!\ell}
\newcommand{\sF}{\mathsf{F}}
\newcommand{\sI}{\mathsf{I}}
\newcommand{\bj}{\mathbf{j}}
\newcommand{\cO}{{\mathcal O}}
\newcommand{\cP}{\mathcal{P}}

\newcommand{\vp}{\varphi}
\newcommand{\cS}{{\mathcal S}}

\newcommand{\srg}{\text{srg}}
\newcommand{\ST}{{\mathcal{ST}}}

\DeclareMathOperator\Aut{Aut}

\DeclareMathOperator\spn{span}
\DeclareMathOperator\Stab{Stab}
\DeclareMathOperator\supp{Supp}

\newtheorem{theorem}{Theorem}[section]
\newtheorem{proposition}[theorem]{Proposition}
\newtheorem{lemma}[theorem]{Lemma}
\newtheorem{corollary}[theorem]{Corollary}
\newtheorem{definition}{Definition}[section]
\newtheorem{example}{Example}[section]

\journal{Journal}

\begin{document}

\begin{frontmatter}

\title{Continuous Time Quantum Walks on Graphs: Group State Transfer} 

\author[wpimainaddress,drexeladdress]{Luke C.~Brown}
\author[wpimainaddress]{William J.~Martin\corref{correspondingauthor}} 
\author[wpimainaddress]{Duncan Wright}

\cortext[correspondingauthor]{Corresponding author. \newline 
\indent \ \ E-mail address: \href{mailto:martin@wpi.edu}{martin@wpi.edu} (W.J.\ Martin).}

\address[wpimainaddress]{Department of Mathematical Sciences, Worcester Polytechnic Institute, Massachusetts, United States}
\address[drexeladdress]{Department of Mathematics, Drexel University, Pennsylvania, United States}

\begin{abstract}
We introduce the concept of group state transfer on graphs,  summarize its relationship to other
concepts in the theory of quantum walks, set up  a basic theory,  and discuss examples.

Let $X$ be a graph with adjacency matrix $A$ and consider quantum walks on the vertex set $V(X)$ governed by the continuous time-dependent
unitary transition operator $U(t)= \exp(itA)$. For $S,T\subseteq V(X)$, we says $X$ admits ``group state transfer'' from $S$ to $T$ at time
$\tau$ if the submatrix of $U(\tau)$ obtained by restricting to columns in $S$ and rows not in $T$ is the all-zero matrix.  As a generalization
of perfect state transfer, fractional revival and periodicity,  group state transfer satisfies  natural monotonicity and transitivity properties. Yet 
non-trivial group state transfer is still rare; using a compactness argument, we prove that bijective group state transfer (the optimal case where $|S|=|T|$)
is absent for almost all $t$. Focusing on this bijective case, we obtain a structure theorem, prove that bijective group state transfer is ``monogamous'', and study
the relationship between the projections of $S$ and $T$ into each eigenspace of the graph. 

Group state transfer is obviously preserved by graph automorphisms and this gives us information about the relationship between the setwise stabilizer of 
$S\subseteq V(X)$ and the stabilizers of naturally defined subsets obtained by spreading $S$ out over time and crudely reversing this process. These
operations are sufficiently well-behaved to give us a topology on $V(X)$ which is likely to be simply the topology of subsets for which bijective group
state transfer occurs at that time.  We illustrate non-trivial group state transfer in bipartite graphs with integer eigenvalues, in joins of graphs, and in symmetric 
double stars. The Cartesian product allows us to build new examples from old ones.
\end{abstract}


\end{frontmatter}


\section{Introduction}
\label{Sec:intro}

Theoretical investigations in quantum computing and quantum information theory have given rise to a number of interesting questions in algebraic graph
theory and nearby areas of combinatorics. Quantum walks on graphs, in particular, seem both fundamental to our 
understanding of how to control the evolution of finite-dimensional quantum mechanical systems and quite amenable to study using the standard
tools of spectral graph theory. Since their introduction in 1998 by Farhi and Gutman \cite{FarGut} as a powerful alternative to classical Markov random 
processes, continuous time quantum walks on graphs and weighted  graphs have received much attention as researchers 
attempt to understand the potential advantages of quantum computation over classical computation. While  Farhi and Gutman  allowed for a sparse real 
Hamiltonian expressible as a sum of Hamiltonians each acting on a 
limited number of underlying qubits,  Childs proved in 2006  that we may restrict attention to Hamiltonians that are simply
adjacency matrices of graphs having maximum degree three and still efficiently simulate any quantum circuit  \cite{Childs}. 

With the path on two vertices as a  classical motivating example \cite{Bose}, Christandl, et al.\ \cite{Christandl} first demonstrated perfect 
quantum state transfer (PST) between vertices at  arbitrary distance $d$ using the product of $d$ such paths to obtain the $d$-cube. 
Graph theorists specializing in spectral techniques soon  developed a theory around such questions (see \cite{God-survey}), showing that 
perfect state transfer is quite rare.  Attention then broadened to include closely related phenomena such as 
periodicity and fractional revival as well as approximations such as pretty good state transfer \cite{VinZhed,God-survey}, among other 
interesting behavior of quantum walks on graphs such as uniform mixing. With path-length distance between vertices as a reasonable surrogate for physical 
distance between components in an implementation of  a quantum system, the hope of finding perfect state transfer between vertices 
far apart in a relatively small unweighted graph seems to have been dashed. Perfect state transfer is not only rare, but the number of 
vertices must grow at least in proportion to the cube of the distance
between the endpoints (and possibly at a much larger rate) \cite{Cout-grow}.

\paragraph{Overview of the paper}
The present work is an outgrowth of the undergraduate senior thesis \cite{LukeMQP}\footnote{The Major Qualifying Project (MQP) at Worcester Polytechnic 
Institute is a campus-wide capstone requirement of all undergraduates.} of the first author (LCB), completed in April 2019 under the 
supervision of the second author (WJM). 
In this paper, claiming no physical motivation, we introduce ``group state transfer'' by which any initial state supported on one set $S$ of vertices
is carried to some state supported on another set $T$. In full generality, group state transfer is ubiquitous: every graph $X$ admits such state transfer 
from the empty set to any subset of vertices and from any set of vertices to the entire vertex set $V(X)$. We call these cases ``trivial''. In Lemma \ref{Lem:basics},
we see how group state transfer behaves with respect to intersections, unions, complements, and time reversal. If $X$ admits group
state transfer from $S$ to $T$ at time $\tau$ then, at time $\tau$, $X$ admits group state transfer from any subset of  $S$ to any superset of $T$. 
This naturally leads (Section \ref{Sec:GST})  to a partial order on such pairs with maximal pairs of particular interest. A compactness argument is
used in Lemma \ref{Lem:usuallytrivial} to show that for all but finitely many values of $\tau$ in any finite interval $[t_0,t_1]$, the only
maximal elements are the trivial ones $(\emptyset,\emptyset)$ and $(V(X),V(X))$. In most strongly regular graphs, only trivial situations arise (Proposition
\ref{Prop:srg}). 

The fundamental inequality  $|S|\le |T|$ in Lemma \ref{Lem:SleqT} can be viewed as an entropy bound and we focus on bijective group 
state transfer, where $|S|=|T|$ in Section \ref{Sec:S=T}.  Using  Lemma \ref{Lem:usuallytrivial}, we prove (Theorem \ref{Thm:monogamous}) that
bijective group state transfer is ``monogamous'' in the sense that, aside from $S$ itself, a set $S$ can be transferred to at most one other
vertex subset of the same size.  Whenever we have bijective group state transfer from $S$ to some other set at time $\tau$, we have group
state transfer from $S$ to itself at time $2\tau$ --- i.e., $S$ is ``periodic at $2\tau$''. Godsil showed that the complement of a periodic set is
again periodic; we show that the collection of vertex subsets periodic at time $\tau$ is closed under intersection and union.
A fundamental restriction on perfect state transfer is the idea of ``parallel vertices'' \cite[Section 6.5]{CoutGod}. Analogous to this, we show
in Lemma \ref{Lem:parallel}  that, if $X$ admits bijective group state transfer from $S$ to $T$ and $E_r$ is any primitive idempotent of the
adjacency algebra of $X$, then there is an $|S|\times |S|$ unitary matrix mapping the columns of $E_r$ indexed by $S$ to the columns 
of $E_r$ indexed by $T$.  

Given a set $S$ of vertices and a time $t$, there are natural targets $R=\sI(S,-t)$ and $T=\sF(S,t)$   for group state transfer to and from $S$, respectively.
In Theorem \ref{Thm:topology},  we consider these maps $\sI(\cdot,\cdot)$ and $\sF(\cdot,\cdot)$ and their composition, which is a closure operation
on vertex subsets giving us a topology on $V(X)$ at any time $t$ (Corollary \ref{Cor:topology}).  This leads into some results in Section \ref{Sec:Aut} 
revealing how group state transfer behaves with respect to the automorphism group of the graph $X$.  

Turning toward examples, Section \ref{Sec:prod}
explores the Cartesian product and join of two graphs. In Proposition  \ref{Prop:prod}, we show that if graph $X$ admits group
state transfer from $S$ to $T$ at time $\tau$ and graph $Y$ admits group state transfer from $S'$ to $T'$ at time $\tau$, then the Cartesian 
product $X \square Y$ admits group state transfer from $S \times S'$ to $T \times T'$ at time $\tau$. In a  simple reformulation of work of 
Coutinho and Godsil \cite{CoutGod}, we find non-trivial group state transfer from $V(X)$ to itself in any join $X+Y$
(Proposition  \ref{Prop:join}). In  Section \ref{Sec:eg}, we list some further examples. For instance, in any bipartite graph $X$ whose eigenvalue 
ratios are all odd integers, we see group state transfer from one bipartite half to the other. Also in Theorem \ref{Thm:bipartite}, we see
periodicity on each bipartite half under weaker conditions. Periodicity is also shown in the
symmetric double star in Proposition \ref{Prop:doublestar}.  We finish the paper with a few more examples and a list of open problems.

\section{Preliminaries}
\label{Sec:prelim}

Throughout, $X=(V(X),E(X))$ is a finite simple graph on $n$ vertices with adjacency matrix $A$. For simplicity, we will sometimes write $V(X)=\{1,\ldots,n\}$. When
$a$ and $b$ are joined by an edge, we write $a\sim b$ or ($ab\in E(G)$) and we use $X(a)=\{ b \in V(X) \mid a \sim b\}$ to denote the neighborhood of $a$ in $X$.
The distance between $a$ and $b$ in $X$, denoted $\partial(a,b)$, is the length of a shortest path joining the two.

The unitary time-dependent transition operator $U(t)=U_X(t)$ is given by
$$ U(t) = \exp( itA) = \sum_{k=0}^\infty \frac{ (it)^k }{ k! } A^k  $$
where $t$ is any real number.
As shown, for example, by Coutinho and Godsil in their text \cite{CoutGod}, the spectral decomposition of $A$ carries over to a useful expression for $U(t)$. Throughout, we suppose that graph $X$ has $d+1$ distinct eigenvalues $\theta_0 > \theta_1 > \cdots > \theta_d$. We denote by $E_r$ the
matrix representing orthogonal projection onto the  eigenspace belonging to $\theta_r$, $V_r =\left\{ \vp \in \cx^n \  \middle| \  A \vp = \theta_r \vp \right\}$. Then we have
$ A = \sum_{r=0}^d \theta_r E_r$ where the various projections sum to the identity:   $\sum_{r=0}^d E_r = I$. 
This gives \cite[Section~1.5]{CoutGod}
\begin{equation}
\label{Eqn:U(t)}
 U(t) =  \sum_{r=0}^d    e^{it \theta_r} E_r ~ .
 \end{equation}

\section{Group state transfer and a partial order on subset pairs}
\label{Sec:GST}

We now give the central definition of this paper. We say graph $X$ admits group state transfer from $S \subseteq V(X)$ to $T\subseteq V(X)$  at
time $\tau$ if the evolution operator $U(\tau)$ carries every initial state vector whose support is contained in $S$ to some vector whose support is
contained in $T$.

\begin{definition}
\label{group_state_transfer}
Let $X$ be a graph and let $S,T\subseteq V(X)$. We say that $X$ has $(S,T)$-\emph{group state transfer}, or $(S,T)$-GST, at time 
$\tau\in\re$ if, for all $\psi\in\cx^n$ such that  $\supp \psi\subseteq S$, the vector $\phi=U_X(\tau)\psi$ satisfies $\supp \phi\subseteq T$.
\end{definition}

For $S\subseteq V(X)$, denote by $\sub{S}$ the subspace of $\cx^n$ of vectors whose support is contained in $S$: 
$\sub{S} = \spn \left\{ e_a | a\in S\right\}$. 
For $S,T\subseteq V(X)$, we have $(S,T)$-GST at time $\tau$ if $U(\tau) \sub{S} \subseteq \sub{T}$.

\paragraph{Familiar examples} Trivial examples include $S=\emptyset$ and $T=V(X)$: for any $R \subseteq V(X)$ and for any 
$\tau \in \re$, we have both $(\emptyset,R)$-GST and 
$(R,V(X))$-GST at  time $\tau$. Our definition of group state transfer, while having no direct physical motivation, generalizes some important 
phenomena that have received much attention in  the quantum information theory community recently. The graph $X$ is said to 
be \emph{periodic} at $a$ at time $\tau$ if $X$ has  $(\{a\}, \{a\})$-GST  at time $\tau$ and, for
$b\neq a$, we say that we have \emph{perfect state transfer} ($ab$-PST) from $a$ to $b$ in $X$ at time $\tau$ 
if $X$ has  $(\{a\}, \{b\})$-GST  at time $\tau$.  The graph $X$ has 
\emph{fractional revival} on $S=\{a,b\}$ at time $\tau$ if $X$ has $(\{a\},S)$-GST at time $\tau$. We use the term \emph{proper fractional revival} when
this holds with  $U(\tau)_{a,b} \neq 0$. (I.e.,  $(\{a\}, \{a,b\})$-GST occurs at time $\tau$ but $(\{a\}, \{a\})$-GST does not.)  It is already 
known that, if $X$ has $(\{a\},\{a,b\})$-GST  at time  $\tau$ then either $a$ is periodic or $X$ has $(\{b\},\{a,b\})$-GST at time $\tau$; see, e.g., Lemma 9.9.1 in \cite{CoutGod}. So $(\{a\}, \{a,b\})$-GST at time $\tau$ implies either that $X$ is periodic at $a$, PST occurs from $a$ to $b$, or we have proper fractional 
revival on $\{a,b\}$ in $X$ (all at time $\tau$). For $S\subseteq V(X)$, the set $S$ is a  \emph{periodic subset} \cite[Section~9.6]{CoutGod} 
if $X$ has $(S,S)$-GST at at some 
time $\tau$ (in which case we say $S$ is \emph{periodic at time} $\tau$)\footnote{Note that, in \cite{CoutGod},  a graph $X$ is said to be 
``periodic'' at time $\tau$ if $U(\tau)$ is a  diagonal matrix; that is, every subset of $V(X)$  is periodic at time $\tau$.}. 

\paragraph{Basic results} We begin with a number of elementary observations that already impose a good deal of structure on the group state transfer 
phenomenon.

\begin{lemma}
\label{Lem:basics}
Let $X$ be a  simple undirected graph. Then
\begin{itemize}
\item[(a)] $X$ admits $(S, V(X))$-GST at time $\tau$ for all $S\subseteq V(X)$ and all times $\tau$;
\item[(b)] $X$ admits $(\emptyset, T)$-GST at time $\tau$ for all $T\subseteq V(X)$ and all times $\tau$.
\item[(c)] $X$ has $(S,T)$-GST at time $\tau$ if and only if $X$ has $( \{a\},T)$-GST at time $\tau$ for every $a\in S$;
\item[(d)] if $S' \subseteq S$ and $T \subseteq T'$ and $(S,T)$-GST occurs at time $\tau$, then $(S',T')$-GST also occurs at time $\tau$;
\item[(e)] if, at time $\tau$, graph $X$ has $(S_1,T_1)$-GST and $(S_2,T_2)$-GST, then $X$ has both $(S_1\cap S_2, T_1\cap T_2)$-GST and 
$(S_1\cup S_2, T_1\cup T_2)$-GST at time $\tau$;
\item[(f)]  if $X$ has $(R,S)$-GST at time $\sigma$  and $X$ has $(S,T)$-GST at time $\tau$,  then  $X$ has $(R,T)$-GST at time $\sigma+\tau$;
\item[(g)] $X$ has $(S,T)$-GST at time $\tau$ if and only if $X$ has $(V(X)\setminus T, V(X)\setminus S)$-GST at time $\tau$;
\item[(h)]  $X$ has $(S,T)$-GST at time $\tau$ if and only if $X$ has $(S,T)$-GST at time $-\tau$.
\end{itemize}
\end{lemma}

\begin{proof}
Parts \textit{(a)} and \textit{(b)} are vacuous. For part \textit{(d)}, we simply observe that, if $U(\tau) \sub{S} \subseteq \sub{T}$, then 
$U(\tau) \sub{S'} \subseteq \sub{T'}$
since $S' \subseteq S$, $T \subseteq T'$ give $\sub{S'} \subseteq \sub{S}$ and $\sub{T}\subseteq \sub{T'}$, respectively. Part \textit{(e)}: suppose
$\varphi \in \sub{S_1\cup S_2} = \sub{S_1} + \sub{S_2}$. Then $U(\tau)\varphi \in  \sub{T_1} + \sub{T_2} =  \sub{T_1\cup T_2}$. (The preservation
of intersections is proved in a similar manner.) Now \textit{(c)} follows from \textit{(d)} and \textit{(e)}.   Part \textit{(f)} is also straightforward. 
Part \textit{(g)} follows from the fact that 
$U(t)$ is a symmetric matrix. Since $U(-\tau)=U(\tau)^{-1} = \overline{U(\tau)}$,  we see that $U(\tau)$ and $U(-\tau)$ have precisely the same set of all-zero
submatrices.
\end{proof}

\begin{example}
\label{Ex:cube}
Suppose graph $X$ admits $a_i b_i$-PST at time $\tau$ for $i=1,\ldots,\ell$. Then, with $S=\{a_1,\ldots, a_\ell\}$ and $T=\{b_1,\ldots,b_\ell\}$,
$X$ admits $(S,T)$-GST at  $\tau$. For instance,  the $d$-cube has PST at time $\pi/2$ from any vertex to its antipode. Let $S\subseteq V(X)$ and 
choose $T$ to consist of  the antipodes of the elements of $S$;  this provides us 
examples with $|S|=|T|$ taking any value up to $|V(X)|=n$ when $X$ is the $d$-cube.
\end{example}

\begin{lemma}
\label{Lem:SleqT}
If graph $X$ has $(S,T)$-GST at $\tau$, then $|S|\leq |T|$.
\end{lemma}

\begin{proof}
Since $U(\tau)$ is injective and $U(\tau) \sub{S} \subseteq \sub{T}$, we have $\dim \sub{S} \le \dim \sub{T}$.
\end{proof}

\paragraph{The State transfer poset} We now introduce the state transfer poset of a graph $X$. Writing $\cP$ for the power set of $V(X)$, 
$$ \cP = \cP (V(X))= \left\{ S \middle| S \subseteq V(X) \right\}, $$
we begin with the poset $( \cP \times \cP, \preceq)$ with partial order 
relation $(S,T) \preceq (S',T')$ if $S\subseteq S'$ and $T' \subseteq T$. For each time $t$, the  \emph{state transfer poset} of $X$ at time 
$t$ is the subposet of this partially ordered set,   depicted in Figure
\ref{Fig:halfcan},  consisting only of those pairs $(S,T)$ for which $X$ has GST at time $t$; this smaller
partially ordered set is denoted $\ST(X,t)$. Note that, at any time $t$,  $\ST(X,t)$ contains the \emph{trivial} pairs $(\emptyset ,T)$ for all $T \subseteq V(X)$ and 
$(S,V(X))$ for all $S \subseteq V(X)$ but may otherwise depend on $t$. One may alternatively view this collection
of pairs $(S,T)$ for which $X$ has GST at time $t$ as a down-set (or ``downward closed set'') in the original poset $( \cP \times \cP, \preceq)$.
This is nothing more than the poset formed by the all-zero submatrices of $U(t)$; we have $(S,T)$-GST at time $t$ precisely when the submatrix of 
$U(t)$ obtained by restricting to rows indexed by elements of $V(X)\setminus T$ and columns indexed by elements of $S$ has all entries zero.

In Figure \ref{Figure:ST(K_2)}, we give the state transfer poset for the path on two vertices $X=K_2$ at time $\tau=\pi/2$.

\definecolor{myblue}{RGB}{200,200,240}

\begin{figure}  
\begin{center}
\begin{tikzpicture}
\def \exsh {-1.27*0.5*2}
\def \eysh {1.27*0.866*2}
\draw[myblue,fill=myblue] (0.58,-1.13) -- (3.95,1.15) -- (0.58+\exsh,-1.13+\eysh) -- (0.58,-1.13);
\draw[rotate=30,fill=myblue]  (0,0) ellipse (0.625 and 1.27);
\draw[rotate=30]  (4,0.29) ellipse (0.625 and 1.27);
\draw[blue,fill=myblue] (1.7,1.15) ellipse (2.25 and 0.12);
\draw[rotate=30] (0,0) ellipse (0.625 and 1.27);
\draw (0.58,-1.13) -- (4.01,1.19);
\draw (0.59+\exsh,-1.13+\eysh) -- (3.98+\exsh,1.17+\eysh);
\node[red] (S0TV) at  (1.58,-1.13)  {$\scriptscriptstyle{(\emptyset,V(X))}$};
\node (SVT0) at  (3.01+\exsh,1.17+\eysh) {$\scriptscriptstyle{(V(X),\emptyset)}$};
\node[red] (S0T0) at (4.6,1.15) {$\scriptscriptstyle{(\emptyset,\emptyset)}$};
\node[red] (SVTV) at (-1.5,1.12) {$\scriptscriptstyle{(V(X),V(X))}$};
\node[red] (PST) at (2.8,1.15) {\tiny{PST}};
\node[rotate=30]  (Left) at (0,0) {$\scriptscriptstyle{  \cP \!\times \!\{V(X)\}}$};
\node[rotate=30]  (Right) at (3.3,2.29) {$\scriptscriptstyle{  \cP \times \{\emptyset\}}$}; 
\node[rotate=35] (rev) at (2.5,-0.2) {\tiny{reverse inclusion}};
\end{tikzpicture} 
\end{center}
\caption{The partially ordered set on $(V(X)\times V(X), \preceq)$ with $(S',T')\preceq(S,T)$ when $S'\subseteq S$ and $T' \supseteq T$. 
Here, $\cP$ denotes the power set of $V(X)$ ordered by containment. The blue region indicates
pairs $(S,T)$ with $|S|\le |T|$ and necessarily contains $\ST(X,\tau)$. Ideal GST occurs at the upper boundary of the blue region, with PST as a special case.
\label{Fig:halfcan}}
\end{figure}
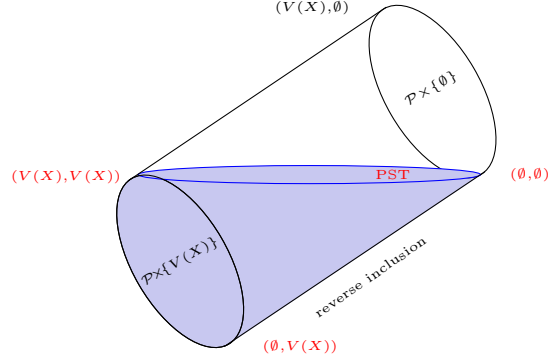  

\paragraph{The extremal case} Let us say that $X$ has \emph{maximal group state transfer} from $S$ to $T$ 
at time $\tau$ if $X$ has $(S,T)$-GST at $\tau$ and, whenever 
$X$ has $(S',T')$-GST at $\tau$ for $S\subseteq S'$ and $T'\subseteq T$, $S'=S$ and $T'=T$. Focusing on a more rare situation, we say $X$ has 
\emph{bijective group state transfer} from $S$ to $T$ at time $\tau$ if $X$ has $(S,T)$-GST at $\tau$ and $|S|=|T|$. Lemma \ref{Lem:SleqT} tells us
that  bijective implies maximal. Given $S\subseteq V(X)$, the maximal element of $\ST(X,t)$ of the form $(S,T)$ is $(S,\sF(S,t))$ where
$$ \sF(S,t) = \left\{ a \in V(X) \  \middle| \ (\exists \varphi \in \sub{S}) ( e_a^\top U(t) \varphi \neq 0) \right\}; $$
that is, $X$ has $(S,T)$-GST at time $t$ if and only if $T \supseteq \sF(S,t)$.  \label{Sec:sF}

\paragraph{Smallest non-trivial elements of the poset}
The most common (and least interesting) case of non-trivial GST (i.e., where $S\neq \emptyset$ and $T\neq V(X)$)  occurs where $U(\tau)$ has some
entry equal to zero: $X$ exhibits $( \{a\}, V(X) \setminus \{b\})$-GST at time $\tau$ if and only if $U(\tau)_{b,a}=0$.  Even this fails almost everywhere.

\begin{lemma}
\label{Lem:usuallytrivial}
Assume $X$ is a connected graph. In any interval $[t_0,t_1]$ of finite 
length, there are only finitely many $t$ for which $\ST(X,t)$ contains non-trivial pairs.
\end{lemma}

\begin{proof}
We need only show that, for $a,b\in V(X)$, $U(t)_{b,a} =0$ for at most finitely many values of $t\in [t_0,t_1]$. Assume not. By compactness, there 
exists a convergent sequence $\{t_n\}_{n=1}^\infty$ of  values all satisfying $U(t_n)_{b,a} =0$. Define
$$ f(t) = \sum_{r=0}^d e^{i \theta_r t} (E_r)_{b,a} ~ .$$
Then $f(t)$ is analytic and $f(t_n)=0$ for all $n$. So, defining $t^* = \lim_{n\rightarrow \infty} t_n$, we obtain $f(t^*)=0$ by continuity. Similarly, every derivative
of $f$ is zero at $t^*$. Since $f$ is analytic, it must be the zero function. But we are assuming that $X$ is connected, so some
element of the adjacency algebra $\sub{A}$ has a nonzero value in its $(b,a)$-position. Since the $\theta_r$ are distinct, there is some $\epsilon >0$ for which 
$U(\epsilon)$ has $d+1$ distinct eigenvalues. Therefore the set $\{ U(t) \mid t\in \re\}$, closed under multiplication, generates $\sub{A}$ and there must be 
some time $t$ at which $U(t)_{b,a} \neq 0$, giving us the desired contradiction.
\end{proof}

\begin{center}  
\begin{figure}[b] 
\resizebox{6in}{!}{\begin{tikzpicture}
\node[red] (ST12) at (0,0) {$\scriptscriptstyle{(\emptyset,\{1,2\})}$};
\node[red] (ST1) at (-0.75,1) {$\scriptscriptstyle{(\emptyset,\{1\})}$};
\node[red] (ST2) at (0.75,1) {$\scriptscriptstyle{(\emptyset,\{2\})}$};
\node[red] (ST) at (0,2) {$\scriptscriptstyle{(\emptyset,\emptyset)}$};
\node[red] (S1T12) at (-3,1.85) {$\scriptscriptstyle{(\{1\},\{1,2\})}$};
\node (S1T1) at (-3.75,2.85) {$\scriptscriptstyle{(\{1\},\{1\})}$};
\node[red] (S1T2) at (-2.25,2.85) {$\scriptscriptstyle{(\{1\},\{2\})}$};
\node (S1T) at (-3,3.85) {$\scriptscriptstyle{(\{1\},\emptyset)}$};
\node[red] (S2T12) at (3.5,1.85) {$\scriptscriptstyle{(\{2\},\{1,2\})}$};
\node[red] (S2T1) at (2.75,2.85) {$\scriptscriptstyle{(\{2\},\{1\})}$};
\node (S2T2) at (4.25,2.85) {$\scriptscriptstyle{(\{2\},\{2\})}$};
\node (S2T) at (3.5,3.85) {$\scriptscriptstyle{(\{2\},\emptyset)}$};
\node[red] (S12T12) at (-0.0,3.7) {$\scriptscriptstyle{(\{1,2\},\{1,2\})}$};
\node (S12T1) at (-0.75,4.7) {$\scriptscriptstyle{(\{1,2\},\{1\})}$};
\node (S12T2) at (0.75,4.7) {$\scriptscriptstyle{(\{1,2\},\{2\})}$};
\node (S12T) at (-0.0,5.7) {$\scriptscriptstyle{(\{1,2\},\emptyset)}$};
\draw[black,thin] (ST12) -- (ST1);   \draw[black,thin] (ST12) -- (ST2);   \draw[black,thin] (ST1) -- (ST);   \draw[black,thin] (ST2) -- (ST);    
\draw[black,thin] (S1T12) -- (S1T1);   \draw[black,thin] (S1T12) -- (S1T2);   \draw[black,thin] (S1T1) -- (S1T);   \draw[black,thin] (S1T2) -- (S1T); 
\draw[black,thin] (S2T12) -- (S2T1);   \draw[black,thin] (S2T12) -- (S2T2);   \draw[black,thin] (S2T1) -- (S2T);   \draw[black,thin] (S2T2) -- (S2T); 
\draw[black,thin] (S12T12) -- (S12T1);   \draw[black,thin] (S12T12) -- (S12T2);   \draw[black,thin] (S12T1) -- (S12T);   \draw[black,thin] (S12T2) -- (S12T); 
\draw[orange,thin] (ST12) -- (S1T12);   \draw[orange,thin] (ST1) -- (S1T1);   \draw[orange,thin] (ST2) -- (S1T2);   \draw[orange,thin] (ST) -- (S1T);    
\draw[orange,thin] (S2T12) -- (S12T12);   \draw[orange,thin] (S2T1) -- (S12T1);   \draw[orange,thin] (S2T2) -- (S12T2);   \draw[orange,thin] (S2T) -- (S12T);  
\draw[green,thin] (ST12) -- (S2T12);   \draw[green,thin] (ST1) -- (S2T1);   \draw[green,thin] (ST2) -- (S2T2);   \draw[green,thin] (ST) -- (S2T);    
\draw[green,thin] (S1T12) -- (S12T12);   \draw[green,thin] (S1T1) -- (S12T1);   \draw[green,thin] (S1T2) -- (S12T2);   \draw[green,thin] (S1T) -- (S12T);  
\end{tikzpicture}  
\begin{tikzpicture}
\node (ST12) at (0,0) {$\scriptscriptstyle{(\emptyset,\{1,2\})}$};
\node (ST1) at (-0.75,1) {$\scriptscriptstyle{(\emptyset,\{1\})}$};
\node (ST2) at (0.75,1) {$\scriptscriptstyle{(\emptyset,\{2\})}$};
\node (ST) at (0,2) {$\scriptscriptstyle{(\emptyset,\emptyset)}$};
\node (S1T12) at (-3,1.85) {$\scriptscriptstyle{(\{1\},\{1,2\})}$};
\node (S1T2) at (-2.25,2.85) {$\scriptscriptstyle{(\{1\},\{2\})}$};
\node (S2T12) at (3.5,1.85) {$\scriptscriptstyle{(\{2\},\{1,2\})}$};
\node (S2T1) at (2.75,2.85) {$\scriptscriptstyle{(\{2\},\{1\})}$};
\node (S12T12) at (-0.0,3.7) {$\scriptscriptstyle{(\{1,2\},\{1,2\})}$};
\draw[blue,thin] (ST12) -- (ST1);   \draw[blue,thin] (ST12) -- (ST2);   \draw[blue,thin] (ST1) -- (ST);   \draw[blue,thin] (ST2) -- (ST);   
 \draw[blue,thin] (S1T12) -- (S1T2);   \draw[blue,thin] (S2T12) -- (S2T1);   
\draw[blue,thin] (ST12) -- (S1T12);    \draw[blue,thin] (ST2) -- (S1T2);   
\draw[blue,thin] (S2T12) -- (S12T12);  
\draw[blue,thin] (ST12) -- (S2T12);   \draw[blue,thin] (ST1) -- (S2T1);   
\draw[blue,thin] (S1T12) -- (S12T12);   
\end{tikzpicture} }
\caption{The poset  $(  \cP (V(K_2) )\times  \cP (V(K_2)), \preceq)$  on the left (reverse inclusion highlighted in black) 
with the subposet identified in red giving us  the state transfer poset $\ST(K_2,\frac{\pi}{2})$ for $K_2$ at time $\tau=\pi/2$.
\label{Figure:ST(K_2)}}
\end{figure}
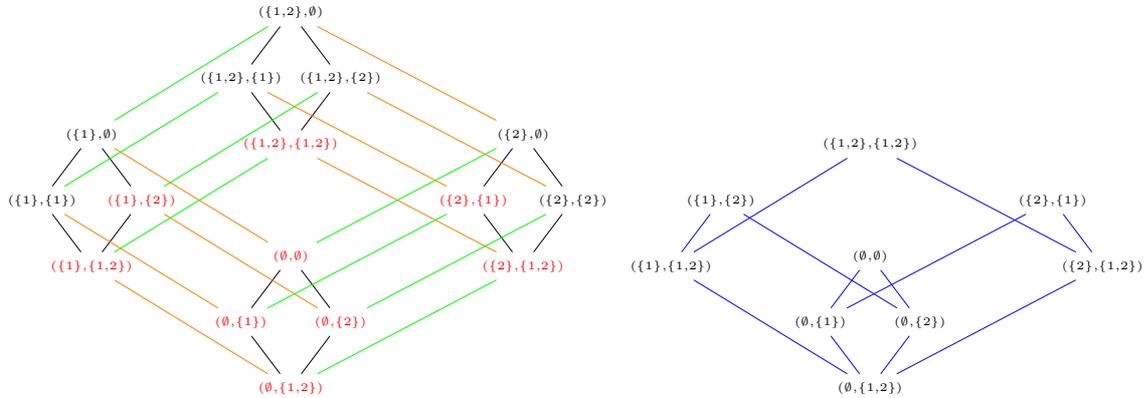 
\end{center}

\paragraph{Strongly regular graphs} For some graphs $X$, there is no value of $t$ in $(0,2\pi)$ for which $\ST(X,t)$ is non-trivial,  as we now illustrate.

A graph $X$ is \emph{strongly regular} with parameters $(\nu,\kappa,\lambda,\mu)$ if $|V(X)|=\nu$ and $|X(a)\cap X(b)|=\kappa,\lambda,\mu$, accordingly, as
$a=b$, $a\sim b$ and $b \not\in \{a\} \cup X(a)$, respectively. We say $X$ is an  $\srg(\nu,\kappa ,\lambda,\mu)$.  
Write $A_0=I$, $A_1=A$ and $A_2=J-I-A$; these form a vector space basis for the adjacency  algebra of $X$.
Standard tools (e.g., \cite[Chapter~10]{GodRoy}) give us the eigenvalues and their multiplicities:
$$ \theta_0 = \kappa, \qquad \theta_1 = \frac12 \left( \lambda - \mu + \sqrt{\Delta} \right), \qquad \theta_2 = \frac12 \left( \lambda - \mu - \sqrt{\Delta} \right) $$
where  $ \Delta = (\mu-\lambda)^2 + 4 (\kappa -\mu)$.
The respective eigenvalue multiplicities for $\theta_1$ and $\theta_2$ are
$$ f = \frac{1}{2}\left( \nu-1 + \frac{(\nu-1)(\mu-\lambda)-2\kappa }{\sqrt{\Delta}} \right), \qquad
g = \frac{1}{2}\left( \nu-1 - \frac{(\nu-1)(\mu-\lambda)-2\kappa }{\sqrt{\Delta}} \right).$$
Except when $f=g$, $\theta_1$ and $\theta_2$ must be integers.  It is well-known that $E_0= \frac{1}{\nu}J$, 
\begin{eqnarray*}
E_1 &=&  \frac{1}{\nu} \left(    f A_0 + \frac{f\theta_1}{\kappa } \, A_1 + \frac{f(1+\theta_1)}{\kappa +1-\nu}   \  A_2   \right),  \\
E_2 &=&  \frac{1}{\nu} \left(    g A_0 + \frac{g\theta_2}{\kappa } \, A_1 + \frac{g(1+\theta_2)}{\kappa +1-\nu}   \  A_2    \right) ~.
\end{eqnarray*}
Choose a base vertex $b\in V(X)$ and define the $\nu \times 3$ matrix $H$ whose columns are $e_b$, $Ae_b$ and $A_2 e_b$.
Since the partition according to distance from $b$ is equitable, we have 
$ A H = H B$
for 
$$ B = \left[ \begin{array}{ccc}   0 & \kappa  &  0 \\   1 &  \lambda & \kappa -1-\lambda  \\   0  &  \mu  & \kappa -\mu \end{array} \right].  $$
Omitting details, we find that  $U(t) = \exp( it A) = e^{it \kappa } E_0 + e^{it \theta_1} E_1 + e^{it \theta_2} E_2$
satisfies $U(t) H = H U'(t)$  where
$ U'(t) = e^{it \kappa } F_0 + e^{it \theta_1} F_1 + e^{it \theta_2} F_2  $ 
with
$$ F_0 = \frac{1}{\nu} \left[ \begin{array}{ccc}  1 & \kappa  & \nu-1-\kappa  \\  1 & \kappa  & \nu-1-\kappa  \\  1 & \kappa  & \nu-1-\kappa   \end{array} \right], \qquad 
F_1 = \frac{f}{\nu} \left[ \begin{array}{ccc}  1 & \theta_1 & -1-\theta_1 \\
\frac{\theta_1}{\kappa } & \frac{ \kappa + \lambda \theta_1 + \mu (-1-\theta_1) }{\kappa } &  \frac{-\theta_1- \kappa  - \lambda \theta_1 + \mu (1+\theta_1) }{\kappa } \\
\frac{-1-\theta_1}{\nu-1-\kappa } & \frac{ -\theta_1- \kappa  - \lambda \theta_1 + \mu (1+\theta_1)   }{\nu-1-\kappa } &  \frac{ 2\theta_1+1+ \kappa  + \lambda \theta_1 + \mu (-1-\theta_1)   }{\nu-1-\kappa }
\end{array} \right] $$
and
$$ 
F_2 = \frac{g}{\nu} \left[ \begin{array}{ccc}  1 & \theta_2 & -1-\theta_2 \\
\frac{\theta_2}{\kappa } & \frac{ \kappa + \lambda \theta_2 + \mu (-1-\theta_2) }{\kappa } &  \frac{-\theta_2- \kappa  - \lambda \theta_2 + \mu (1+\theta_2) }{\kappa } \\
\frac{-1-\theta_2}{\nu-1-\kappa } & \frac{ -\theta_2- \kappa  - \lambda \theta_2 + \mu (1+\theta_2)   }{\nu-1-\kappa } &  \frac{ 2\theta_2+1+ \kappa  + \lambda \theta_2 + \mu (-1-\theta_2)   }{\nu-1-\kappa }
\end{array} \right] .$$
If the system is in initial state $e_b$ at time zero, then at time $t$, the state of the system is given by 
$U(t) e_b = H U'(t) \left[ \begin{array}{ccc} \scriptstyle{1} & \scriptstyle{0}  & \scriptstyle{0} \end{array} \right]^\top$. 
So $U(t) e_b$ is constant on the neighbors of $b$ and and on the non-neighbors of $b$. Define
\begin{eqnarray*}
h_0(t) &=& e^{i\kappa t} + f e^{i\theta_1t} + g e^{i\theta_2t}  \\
h_1(t) &=& e^{i\kappa t} + (f\theta_1/\kappa ) e^{i\theta_1t} + (g\theta_2/\kappa ) e^{i\theta_2t}  \\
h_2(t) &=& e^{i\kappa t} + f(1+\theta_1)/(\kappa +1-\nu) e^{i\theta_1t} + g(1+\theta_2)/(\kappa +1-\nu) e^{i\theta_2t} ~ .
\end{eqnarray*}
Then $e_a^\top U(t)e_b= \frac{1}{\nu} h_\delta(t)$ where $\delta \in \{0,1,2\}$ is the distance from $a$ to $b$ in $X$. This tells us that 
GST almost never occurs on strongly regular graphs.

\begin{proposition}
\label{Prop:srg}
Let $X$ be a connected strongly regular graph with non-trivial $(S,T)$-GST at time $\tau \in (0,2\pi)$. Then one of the following occurs:
\begin{itemize}
\item[(a)]  $\kappa, \theta_1,\theta_2$ are all integers divisible by some  $D\ge 2$ and $\tau =2\ell\pi/D$ where $\ell$ is an integer, $0 < \ell < D$;
\item[(b)]  $(\nu,\kappa,\lambda,\mu)=(n,n-m,n-2m,n-m)$, $X$ is complete multipartite, 
the complement of a disjoint union of $| \theta_2 |= \frac{n}{m} >2$  complete graphs $K_m$,  $\tau=2\pi \ell/m$;
\item[(c)] $(\nu,\kappa,\lambda,\mu)=(2m,m,0,m)$,  $X$ is complete bipartite and   $\tau = \pi/D$ where $D$ is any positive divisor of $\kappa$;
\item[(d)]  $(\nu,\kappa,\lambda,\mu)=( 4 m+1, 2 m,  m-1, m)$ and $\tau = 2\pi B/\nu$ for some integer $B$ satisfying 
$\cos \left( \pi B \nu^{-1/2} \right) = -1/4m$.
\end{itemize}
\end{proposition}

\begin{proof} 
We have done most of the work already. Part \textit{(c)} is handled in  Theorem \ref{Thm:bipartite} below.

First note that $1+f+g=\nu$; if $|f-g|>1$, then $h_0(t)$ is never zero, by the triangle inequality. Likewise, since
$$ 1 + f \frac{\theta_1}{\kappa} + g \frac{\theta_0}{\kappa} =  1+f \frac{\theta_1+1}{\kappa +1-\nu} + g \frac{\theta_2+1}{\kappa +1-\nu} = 0,$$
we can only have $h_2(t)=0$ when  $e^{i\kappa t} =e^{i\theta_1t}  = e^{i\theta_2t}$, in which case $U(t)e_b = e_b$. These are the only times at which 
$h_1(t)=0$ with the exception of complete multipartite graphs $X = \overline{| \theta_2 | K_m}$ (where $\theta_1 =0$) in which case we obtain 
$(\{b\},V(X) \setminus X(b))$-GST at times  $t=2\pi \ell/m$, maximal for $\ell$ odd. 

In the case $f=g$, it is well-known that the parameters $(\nu,\kappa,\lambda,\mu)$ are as given in case \textit{(d)} with $f=g=2\mu$ and 
$\theta_1,\theta_2 = \frac12 (-1 \pm \sqrt{\nu} )$. To obtain $h_0(\tau)=0$, we must have  
$\tau( \theta_1 - \kappa) + \tau(\theta_2-\kappa)$ an integer multiple of $2\pi$. Writing $\tau=-2\pi B/\nu$,
we need
$$ e^{i \tau( \theta_1 - \kappa) } + e^{i \tau( \theta_1 - \kappa) }  = -1 / \kappa $$
which gives us the condition $\cos \left( \pi B/ \sqrt{4\mu+1} \right) = -1/4\mu$ and no such examples are known.

The only case that remains to consider is 
$(\{b\},V(X)\setminus \{b\})$-GST in the case where $|f-g|=1$.  Aleksandar Juri\v{s}i\'{c} [pers.\ communication] showed that the strongly regular 
graph parameters with $|f-g|=1$ are precisely those in the family
$$ (\nu,\kappa,\lambda,\mu) = ( 4m^2+4m+2, \ 2m^2+m, \ m^2-1, \ m^2) $$
where $m$ is a positive integer. And now a simple parity argument shows $U(t)_{b,b} \neq 0$ for all real $t$. 
\end{proof}   

A regular graph $X$ is \emph{distance-regular} if the partition according to distance from any vertex is an equitable partition (and, hence, all these partitions 
admit the same quotient matrix $B$ \cite{GodShaw}).  Connected strongly regular graphs are precisely the distance-regular graphs of diameter $d=2$. 
The analysis above for strongly regular graphs extends to distance-regular graphs in the following way: if $X$ is a distance-regular graph of diameter
$d$ and $X$ admits $(S,T)$-GST at time $\tau$, then there exist $i_1,\ldots,i_k \in \{0,1,\ldots,d\}$, $k>0$,  for which 
$$T  \supseteq \left\{   v \in V(X) \middle| \left( \exists u\in S, \ 1\le j\le k\right) \left( \partial(u,v)=i_j  \right) \right\}$$
where $\partial(u,v)$ denotes path-length distance between $u$ and $v$ in $X$.

\subsection{Bijective group state transfer}
\label{Sec:S=T}

\paragraph{Block matrices} Let us consider the block structure of $U_X(\tau)$ when $X$ admits $(S,T)$-GST at time $\tau$. For convenience, assume the vertex
set $V(X)=\{1,\ldots,=n\}$ is ordered so that 
$$ S\setminus T = \{ 1,\ldots, n_2\}, \quad I=S\cap T= \{n_1+1,\ldots, n_2\}, \quad T\setminus S = \{n_2+1,\ldots n_3\} $$
where $1 \le n_1 \le n_2 \le n_3 \le n$. Partition the rows and columns accordingly and write 
$$ U(\tau)=U_X(\tau) = \hspace{.2in} 
\begin{tikzpicture}[baseline=(t),decoration=brace]
    \matrix (m) [matrix of math nodes,left delimiter=[,right delimiter={]}] {
    U_{11} & U_{12} & U_{13}& U_{14} \\ 
U_{21} & U_{22} & U_{23}& U_{24} \\ 
U_{31} & U_{32} & U_{33} & U_{34} \\ 
U_{41} & U_{42} & U_{43} & U_{44} \\
    };
    \draw[decorate,transform canvas={xshift=-1.5em},thick] (m-3-1.south west) -- node[left=2pt] (t) {$T$} (m-2-1.north west);
    \draw[decorate,transform canvas={yshift=0.5em},thick] (m-1-1.north west) -- node[above=2pt] {$S$} (m-1-2.north east);
\end{tikzpicture}
 = \left[ \begin{array}{c|c|c|c} 
 0 &  0 & U_{13}& U_{14} \\ \hline
U_{21} & U_{22} & U_{23}& U_{24} \\ \hline
U_{31} & U_{32} & U_{33} & U_{34} \\ \hline 
0  & 0 & U_{43} & U_{44}  \end{array} \right]$$
using the hypothesis of $(S,T)$-GST. Since $U(\tau)$ is a symmetric matrix, we have
$$ U(\tau)  = \left[ \begin{array}{c|c|c|c} 
 0 &  0 & U_{13}&  0 \\ \hline
0 & U_{22} & U_{23}& 0 \\ \hline
U_{31} & U_{32} & U_{33} & U_{34} \\ \hline 
0  & 0 & U_{43} & U_{44}  \end{array} \right]$$
with $U_{31}=U_{13}^\top$,  $U_{32}=U_{23}^\top$, 
$U_{43}=U_{34}^\top$,  and $U_{jj}$ symmetric for $j=2,3,4$.  

\begin{proposition}
\label{Prop:VminusTtoVminusS}
If $X$ has $(S,T)$-GST at time $\tau$, then $X$ has $(V(X)\setminus T, V(X)\setminus S)$-GST at time $\tau$. \qed
\end{proposition}

The Frobenius norm of $U_{31}$ is $n_1$, so the sum of the squared moduli of the 
entries of $U_{13}$ is also $n_1$, giving another proof that $|S|\le |T|$. If $|S|=|T|$, then $U_{j3}=0$ is forced for $j=2,3,4$. So, for $|S|=|T|$, we have
$$ U(\tau)  = \left[ \begin{array}{c|c|c|c} 
 0 &  0 & U_{13}&  0 \\ \hline
0 & U_{22} & 0& 0 \\ \hline
U_{31} & 0 &0 & 0 \\ \hline 
0  & 0 & 0& U_{44}  \end{array} \right]$$ 
with $U_{22}$ and $U_{44}$ symmetric unitary matrices. This gives us the following result\footnote{Godsil [personal communication] studied the case
of a periodic subset (where $S=T$), showing not only that $V(X)\setminus S$ is also periodic but proving that $Q_S U(\tau) Q_S$ belongs to 
the center of the algebra $Q_S \cA Q_S$ where $Q_S = \sum_{a \in S} e_a e_a^\top$ is the diagonal matrix projecting $\cx^n$ orthogonally onto $\sub{S}$.}

\begin{theorem}
\label{Thm:equal_card}
Assume that graph $X$ has $(S,T)$-GST at time $\tau$ and $|S|=|T|$. Write $I=S\cap T$. Then
\begin{itemize}
\item[(a)] $X$ has $(T,S)$-GST at $\tau$;
\item[(b)] $X$ has $(S\setminus I,T\setminus I)$-GST at $\tau$;
\item[(c)] $X$ has $(T\setminus I,S\setminus I)$-GST at $\tau$;
\item[(d)] $I$ is periodic at $\tau$; 
\item[(e)] both $S$ and $T$ are periodic at time $2\tau$;
\item[(f)]  the set $R=V(X) \setminus S\cup T$ is periodic at $\tau$.  \qed
\end{itemize}
\end{theorem}

\begin{corollary} 
\label{Cor:assume_disj}
If $X$ is a graph with $(S,T)$-GST at time $\tau$ such that $|S|=|T|$, then there exist disjoint  $S',T' \subseteq V(X)$ for which $|S'|=|T'|$ and
$X$ has $(S',T')$-GST at time $\tau$. \qed
\end{corollary}

We can now prove that group state transfer is  \emph{monogamous}. 

\begin{theorem}
\label{Thm:monogamous}
If $X$ admits $(S,R)$-GST at time $\sigma$ and $(S,T)$-GST at time $\tau$ with $|R|=|S|=|T|$, then  $R\in \{S,T\}$.
\end{theorem}

\begin{proof}
We first prove that $\sigma$ and $\tau$ must be commensurable  real numbers.  If not, then the set of remainders $\rho_n = n \tau \pmod{\sigma}$ (satisfying
$0\le \rho_n < \sigma$ and $(n\tau - \rho_n)/\sigma \in \ints$) must be infinite.  Re-index to a subsequence of $\ints^+$ if necessary so that,
with $m(n)= (n\tau - \rho_n)/\sigma$, we have $\tau_n = n\tau - m(n) \sigma$ converging to some point $\tau^* \in [0,\sigma)$. Applying Lemma \ref{Lem:basics}\textit{(f)}
and Theorem \ref{Thm:equal_card}\textit{(a,e)}, we find infinitely many distinct times at which $(K,L)$-GST occurs for some $K,L \in \{R,S,T\}$, contradicting
Lemma \ref{Lem:usuallytrivial}. So there must be some $n$ and $n'$ for which $\rho_n = \rho_{n'}$ and we have $\sigma = (n-n')\tau/( m(n) - m(n'))$.

Since $\sigma$ and $\tau$ are commensurable, there exist integers $m,n$ such that $m\sigma=n\tau$ and, without loss of generality, $m$ is odd. At time
$m\sigma$, $X$ admits $(S,R)$-GST and either $(S,S)$-GST or $(S,T)$-GST. Thus $R=S$ or $R=T$.
\end{proof}

\section{Eigenspace geometry}   
\label{Sec:eig}

Let $X$ be a graph with adjacency matrix $A$ and spectral decomposition $A = \sum_{r=0}^d \theta_r E_r$ with $\theta_0,\ldots,\theta_d$ distinct.
The \emph{adjacency algebra} 
$\cA = \spn_\cx \left\{ I, A,A^2, \ldots \right\} = \left\{ \sum_{r=0}^d \alpha_r A^r \ \middle| \ \alpha_0,\ldots, \alpha_d \in \cx\right\}$ of $X$ 
contains $E_0,\ldots,E_d$ as well as $U_X(t)$ for each $t\in \re$. This is properly contained in the  \emph{centralizer algebra}
$ \cC(A) = \left\{ M \in \cx^{n \times n} \middle| MA=AM \right\}$ of $A$. The permutation matrices in $\cC(A)$  are  simply
those representing elements of the automorphism group, $\{ P_\sigma \mid \sigma \in \Aut(X) \}$.

\paragraph{The action of $U(\tau)$ on an eigenspace} 
Suppose $X$ admits $(S,T)$-GST at time $\tau$ with $|S|=|T|$. As in the previous section, write $U(\tau)$ in block form and partition $E_r$ into 
blocks in the same way:

$$ U = U(\tau)  = \left[ \begin{array}{c|c|c|c} 
 0 &  0 & U_{13}&  0 \\ \hline
0 & U_{22} & 0& 0 \\ \hline
U_{31} & 0 &0 & 0 \\ \hline 
0  & 0 & 0& U_{44}  \end{array} \right], \qquad 
E = E_r =
  \left[ \begin{array}{c|c|c|c} 
 E_{11} &  E_{12} & E_{13}& E_{14} \\ \hline
E_{21} & E_{22} & E_{23}& E_{24} \\ \hline
E_{31} & E_{32} & E_{33} & E_{34} \\ \hline 
E_{41}  & E_{42} & E_{43} & E_{44}  \end{array} \right].$$
Abbreviating $e^{i\tau \theta_r } = \lambda_r$,  the equations $E U = U E = \lambda_r E$ give us a system of equations relating 
the various blocks
\begin{center}
\begin{tabular}{rm{.01in}cm{.01in}ccrm{.01in}cm{.01in}c}
$E_{13} U_{31}$ & $=$ & $U_{13} E_{31}$ & $=$ & $\lambda_r E_{11}$  &  \hspace{0.4in} & $E_{11} U_{13}$ & $=$ & $U_{13} E_{33}$ & $=$ & $\lambda_r E_{13}$  \\
$E_{23} U_{31}$ & $=$ & $U_{22} E_{21}$ & $=$ & $\lambda_r E_{21}$  &  \hspace{0.4in} & $E_{21} U_{13}$ & $=$ & $U_{22} E_{23}$ & $=$ & $\lambda_r E_{23}$  \\
$E_{33} U_{31}$ & $=$ & $U_{31} E_{11}$ & $=$ & $\lambda_r E_{31}$  &  \hspace{0.4in} & $E_{31} U_{13}$ & $=$ & $U_{31} E_{13}$ & $=$ & $\lambda_r E_{33}$  \\
$E_{43} U_{31}$ & $=$ & $U_{44} E_{41}$ & $=$ & $\lambda_r E_{41}$  &  \hspace{0.4in} & $E_{41} U_{13}$ & $=$ & $U_{44} E_{43}$ & $=$ & $\lambda_r E_{43}$  \\[.1in]
$E_{12} U_{22}$ & $=$ & $U_{13} E_{32}$ & $=$ & $\lambda_r E_{12}$  &  \hspace{0.4in} & $E_{14} U_{44}$ & $=$ & $U_{13} E_{34}$ & $=$ & $\lambda_r E_{14}$  \\ 
$E_{22} U_{22}$ & $=$ & $U_{22} E_{22}$ & $=$ & $\lambda_r E_{22}$  &  \hspace{0.4in} & $E_{24} U_{44}$ & $=$ & $U_{22} E_{24}$ & $=$ & $\lambda_r E_{24}$  \\ 
$E_{32} U_{22}$ & $=$ & $U_{31} E_{12}$ & $=$ & $\lambda_r E_{32}$  &  \hspace{0.4in} & $E_{34} U_{44}$ & $=$ & $U_{31} E_{14}$ & $=$ & $\lambda_r E_{34}$  \\ 
$E_{42} U_{22}$ & $=$ & $U_{44} E_{42}$ & $=$ & $\lambda_r E_{42}$  &  \hspace{0.4in} & $E_{44} U_{44}$ & $=$ & $U_{44} E_{44}$ & $=$ & $\lambda_r E_{44}$  
\end{tabular}
\end{center}
where we know that both $E$ and $U$ are symmetric and $U$ is unitary. So $U_{13},U_{22},U_{31},U_{44}$ are all unitary. This shows that $S\setminus I$ and $T\setminus I$  are ``parallel'' subsets in the following sense.

\begin{lemma}
\label{Lem:parallel}
Let $X$ be a graph with adjacency matrix $A$ having spectral decomposition $A=\sum_{r=0}^d \theta_r E_r$ with $\theta_0,\ldots,\theta_d$ distinct.  Let
$S,T\subseteq V(X)$ with $|S|=|T|$ having orthogonal projections $Q_S = \sum_{a\in S} e_a e_a^\top$ and $Q_T = \sum_{a\in T} e_a e_a^\top$ onto
$\sub{S}$ and $\sub{T}$, respectively. If $X$ admits $(S,T)$-GST, then, for each $r=0,\ldots,d$, there exists a unitary matrix $N_r$ such that 
$E_r Q_S N_r = E_r Q_T$. In particular $\spn \{ E_r e_a \mid a\in S\} =  \spn \{ E_r e_a \mid a\in T\}$.
\end{lemma}

\begin{proof}
Write $M = \lambda_r^{-1} U_{13}$, so that
$$   \left[ \begin{array}{c}   E_{11}  \\  E_{21} \\ E_{31} \\ E_{41}  \end{array} \right] M = 
    \left[ \begin{array}{c}   E_{13}  \\  E_{23} \\ E_{33} \\ E_{43}  \end{array} \right]  $$
from above. Choose 
$$N' = \left[ \begin{array}{c|c|l|c} 
M& 0 &  0 &  0 \\ \hline
0 & I & 0& 0 \\ \hline
0 & 0 & M^{-1} &  0 \\ \hline 
0  & 0 & 0& I  \end{array} \right] , \quad
P = 
 \left[ \begin{array}{c|c|c|c} 
0&  0 & I &  0 \\ \hline
0 & I  & 0& 0 \\ \hline
I  & 0 & 0&  0 \\ \hline 
0 & 0 & 0& I  \end{array} \right] , 
 \quad \text{so that} \quad
E_r N' =    \left[ \begin{array}{c|c|c|c} 
 E_{13} &  E_{12} & E_{11}& E_{14} \\ \hline
E_{23} & E_{22} & E_{21}& E_{24} \\ \hline
E_{33} & E_{32} & E_{31} & E_{34} \\ \hline 
E_{43}  & E_{42} & E_{41} & E_{44}  \end{array} \right]$$
and $N_r = N'P$ satisfies $E_r Q_S N_r = E_r Q_T$ as desired.
\end{proof}

\section{Discrete topology} 
\label{Sec:top}

\paragraph{Three maps on subsets of vertices} In Section \ref{Sec:sF}, we introduced a time-dependent function $\sF: \cP \rightarrow \cP$ given by 
$$ \sF(S, t) = \left\{ a\in V(X) \ \middle| \ (\exists b\in S) ( e_a^\top U(t) e_b \neq 0) \right\}. $$
Mirroring this, consider 
$$ \sI(S, t) = \left\{ a\in V(X) \ \middle| \   e_a  \in U(t) \sub{S}  \right\}. $$
Immediately, we see that the following are equivalent for $S,T\subseteq V(X)$:
\begin{itemize}
\item $X$ has $(S,T)$-GST at time $\tau$;
\item $\sF(S,\tau) \subseteq T$;
\item $S \subseteq \sI(T,-\tau)$.
\end{itemize}
Now define the $t$-\emph{closure} of $S\subseteq V(X)$ as
$$ \Cl_t(S) = \sI \left( \sF(S,t), -t \right).$$

\begin{theorem}
\label{Thm:topology}
Let $X$ be a graph and let $S,T\subseteq V(X)$. Then, for any $t\in \re$,
\begin{itemize}
\item[(a)] $S \subseteq T$ implies $\sF(S,t) \subseteq \sF(T,t)$;
\item[(b)] $\sF( S \cap T, t)  \subseteq  \sF(S,t) \cap \sF(T,t)$;
\item[(c)] $\sF( S \cup T, t) = \sF(S,t) \cup \sF(T,t)$;
\item[(d)] $S \subseteq T$ implies $\sI(S,t) \subseteq \sI(T,t)$;
\item[(e)] $\sI( S \cap T, t)  = \sI(S,t) \cap \sI(T,t)$;
\item[(f)] $\sI( S \cup T, t) \supseteq  \sI(S,t) \cup \sI(T,t)$;
\item[(g)] $S \subseteq \Cl_t(S)$;
\item[(h)]  $\Cl_t( S \cap T) \subseteq  \Cl_t(S) \cap \Cl_t(T)$;
\item[(i)]  $\Cl_t(S) \cup \Cl_t(T)\subseteq   \Cl_t( S \cup T) $.
\end{itemize}
\end{theorem}

\begin{proof}
The proofs are all elementary. We include proofs of \textit{(b)},  \textit{(f)}, and  \textit{(g)} only. First, we prove part  \textit{(b)}: if $u\in  \sF( S \cap T, t)$ then
there is some $v\in S\cap T$ with $e_u^\top U(t) e_v \neq 0$. Since $v\in S$, $u \in \sF(S,t)$ and since $v\in T$, $u \in \sF(T,t)$. For  \textit{(f)}, take $u\in
 \sI(S,t)$ so that $e_u \in U(t) \sub{S}$, giving $e_u \in U(t) \sub{S}+ U(t) \sub{T}= U(t) \sub{S\cup T}$ and repeat this reasoning with the roles of $S$ and $T$ swapped.
 To prove   \textit{(g)}, take $u\in S$ and set $\vp=U(t)e_u \in  U(t) \sub{S}$. Then $\supp(\vp) \subseteq \sF(S,t)$ giving $\vp \in \sub{ \sF(S,t)}$ which, in
 turn, implies $e_u \in \sI( \sF(S,t),-t)=\Cl_t(S)$. 
\end{proof}

\paragraph{Discrete topology} Let us say that $S\subseteq V(X)$ is \emph{closed at time} $t$  (or simply $t$-\emph{closed}) 
if $S = \Cl_t(S)$ and  \emph{open at time} $t$ if $V(X) \setminus S$ is closed at time $t$. 
Note that any set $S$ for which $X$ admits bijective group state transfer at time $t$ (i.e.,  $|\sF(S,t)|=|S|$) 
is $t$-closed. We do not know if every $t$-closed subset has this property.  If so, then $t$-closed and $t$-open are synonymic. From Lemma \ref{Lem:usuallytrivial},
we know that, for most $t$, we obtain only the indiscrete topology $\{ \emptyset, V(X)\}$  and Example \ref{Ex:cube} illustrates a case where the discrete topology
arises: at time $t=\pi/2$, every vertex subset of the $d$-cube is both $t$-open and $t$-closed.

\begin{corollary}
\label{Cor:topology}
Let $X$ be a graph. At each time $t$, the open sets at time $t$ form a topology on $V(X)$. \qed
\end{corollary}

\begin{proof}
Both $\emptyset$ and $V(X)$ are $t$-closed for all $t$. By parts  \textit{(g)} and  \textit{(h)} of Theorem \ref{Thm:topology}, the intersection of any two $t$-closed
sets is $t$-closed. Now assume $S$ and $T$ are both $t$-closed subsets of $V(X)$. By part \textit{(i)} above, $S \cup T \subseteq   \Cl_t( S \cup T) $. So
consider $u\in   \Cl_t( S \cup T) $. Then 
$$ e_u \in \sI\left( \sF( S\cup T, t), -t \right) = U(-t) \sub{\sF(S,t) \cup \sF(T,t)}= U(-t)  \sub{\sF(S,t)} +  U(-t) \sub{\sF(T,t)} = \sub{S}+\sub{T}=\sub{S\cup T}. $$
\end{proof}

We will show below that the projection map from a Cartesian product of graphs to any individual factor is continuous relative to the two topologies at time $t$. We ask
if there are any other interesting cases of covering maps that are continuous in this sense.

\section{GST and the automorphism group} 
\label{Sec:Aut}

We continue with a graph $X$ on vertex set $V(X)=\{1,\ldots,n\}$ and adjacency matrix $A$. 
Using $\cS_n$ to denote the symmetric group, we denote by $\Aut(X)$ the automorphism group of $X$: if $P_\sigma$ is the 
permutation matrix representing the bijection $\sigma: V(X) \rightarrow V(X)$, then
$\Aut(X) = \{ \sigma \in \cS_n \mid P_\sigma A = A P_\sigma \}$.  For $a\in V(X)$ and $H \le \Aut(X)$, the orbit of $a$ under $H$ will be denoted 
$\cO_H(a) = \{ a^\eta \mid \eta \in H \}$ and, writing $S^\eta = \{ a^\eta \mid a\in S\}$, 
the orbit of $S\subseteq V(X)$ under $H$ will be denote $\cO_H(S) = \{ S^\eta \mid \eta \in H \}$. 
The setwise stabilizer of $S$  is $\Stab(S) = \{ \sigma \in \Aut(X) \mid S^\sigma = S\}$.

\begin{proposition}
Let $X$  be a graph, $S,T\subseteq V(X)$. Assume $X$ admits $(S,T)$-GST at time $\tau$.  Then
\begin{itemize}
\item[(a)]  for  any  $\sigma  \in \Aut(X)$, $X$ admits $(S^\sigma,T^\sigma)$-GST at time $\tau$;
\item[(b)] setting $H=\Stab(X)$ and $T' = \displaystyle{\bigcap_{\eta \in H}} T^\eta$, $X$ admits $(S,T')$-GST at time $\tau$;
\item[(c)] if $|S|=|T|$, then $\Stab(S)=\Stab(T)$.
\end{itemize}
\end{proposition}

\begin{proof}
For part \textit{(a)},  $u \in S^\sigma$ gives $U(\tau) e_u = P_\sigma U(\tau) P_\sigma^{-1} e_u = P_\sigma \psi$ for some $\psi \in \sub{T}$. Now part \textit{(b)}
follows by applying \textit{(a)} to each $\sigma \in H$ and using Lemma \ref{Lem:basics}\textit{(e)}. Part  \textit{(c)} follows from  \textit{(a)} using Theorem \ref{Thm:equal_card} \textit{(a)}.
\end{proof}

Using this, together with Lemma \ref{Lem:basics}\textit{(c)}, we have 

\begin{corollary} If $X$ has $(u,v)$-PST at $\tau$, then $X$ has $(\cO (u),\cO(v))$-GST at $\tau$. $X$ also has $(\cO(v),\cO(u))$-GST at $\tau$, where 
$\cO(u)$ and $\cO (v)$ denote the orbit under any subgroup $H$ of $\Aut(X)$. $\qed$
\end{corollary}

\begin{proposition}
Let $X$  be a graph, $S \subseteq V(X)$; write $R=\sI(S,t)$ and $T=\sF(S,t)$. Then
\begin{itemize}
\item[(a)]    $\Stab(S) \le \Stab(R)$  and  $\Stab(S) \le \Stab(T)$;
\item[(b)]  $|\cO(S)| \ge | \cO(R)|$ and  $|\cO(S)| \ge |\cO(T)| $.
\end{itemize}
\end{proposition}

\begin{proof}
Suppose $\sigma \in \Stab(S)$. For $v\in T$, locate $\psi \in \sub{S}$ with $e_v^\top U(t)\psi \neq 0$. Then $e_{v^\sigma} = P_\sigma e_v$ and
$e_{v^\sigma}^\top U(t) \varphi = e_v^\top U(t) \psi \neq 0$ for $\varphi = P_\sigma \psi \in \sub{S}$ since $\sigma \in \Stab(S)$. This shows $v^\sigma \in T$.
On the other hand, if $v\in R$, then $\varphi = U(t) e_v \in  \sub{S}$ so $U(t) e_v^\sigma =  U(t) P_\sigma e_v = P_\sigma \varphi \in \sub{S}$ since $\sigma$ preserves
$S$. This shows that $\sigma$ preserves $R$.  Part \textit{(b)} now follows by the Orbit-Stabilizer Theorem.
\end{proof}

Lemma \ref{Lem:usuallytrivial} tells us that we almost always have $R=\emptyset$ and $T=V(X)$; in such cases, the above result is vacuous.

\section{Products and joins} 
\label{Sec:prod}

\begin{proposition} 
\label{Prop:prod}
Let $X_1$ and $X_2$ be connected graphs. Assume  that  $X_1$ has $(S_1,T_1)$-GST at time $\tau$ and $X_2$ has $(S_2,T_2)$-GST at
time $\tau$.  Then $X_1\square X_2$ has $(S_1\times S_2,T_1 \times T_2)$-GST at $\tau$, where $\square$ denotes the Cartesian graph product.
\end{proposition}

\begin{proof}
 Let $U_1(t)=U_{X_1}(t)$ and $U_2(t)=U_{X_2}(t)$. We know from \cite[Find page]{CoutGod} that
 $$ U_{X_1 \square X_2} (t ) = U_1(t) \otimes U_2(t). $$
\par Suppose that graph $X_1$ has $(S_1,T_1)$-GST at $\tau$, and graph $X_2$ has $(S_2,T_2)$-GST at $\tau$. If $(a_1,a_2)\in S_1\times S_2$,
then we may write $e_{(a_1,a_2)} = e_{a_1} \otimes e_{a_2}$ and we compute
 $$ U_{X_1 \square X_2} (t ) e_{(a_1,a_2)}  = \left( U_1(t) \otimes U_2(t) \right)  ( e_{a_1} \otimes e_{a_2}) = \left( U_1(t) e_{a_1}  \right) \otimes  
 \left( U_2(t)e_{a_2}  \right) .$$
 Since $U_1(t) e_{a_1}  \in \sub{T_1}$ and  $U_2(t) e_{a_2}  \in \sub{T_2}$, we have $U_{X_1 \square X_2} (t ) e_{(a_1,a_2)}  \in \sub{T_1\times T_2}$.
\end{proof}

As a special case, we have the following, using Lemma \ref{Lem:basics}\textit{(a)}.

\begin{proposition} 
\label{Prop:SxV}
Let $X$ and $Y$ be connected graphs, so that $X$ has $(S,T)$-GST at $\tau$.  Then $X\square Y$ has $(S\times V(Y),T\times V(Y))$-GST at $\tau$. \qed
\end{proposition}

One curious consequence of this is that the projection maps $\pi_i : (a_1,a_2) \mapsto a_i$  from $X \square Y$ are continuous relative to the topologies 
of $t$-open sets for all times $t$.

\paragraph{The join} Let $X_1$ and $X_2$ be connected graphs on disjoint vertex sets and define $X=X_1+X_2$ to be the graph on 
vertex set $V(X)=V(X_1)\cup V(X_2)$
with edge set $E(X) = E(X_1) \cup E(X_2) \cup \{ ab \mid a\in V(X_1), \ b\in V(X_2)\}$. The graph $X$ is the \emph{join} of $X_1$ and $X_2$. Denoting the 
adjacency matrices of the three graphs by $A(X_1)$, $A(X_2)$ and $A(X)$, we have
$$ A(X) = \left[ \begin{array}{c|c}  A(X_1)  & J  \\ \hline {J^\top}^{\phantom{2}} & A(X_2) \end{array} \right] $$
where $J$ is the all ones matrix with $m_1=|V(X_1)|$ rows and $m_2=|V(X_2)|$ columns. In the case that $X_1$ and $X_2$ are regular graphs, 
a basis of eigenvectors for $A(X)$ can be derived from eigenbases for $A(X_1)$ and $A(X_2)$ as shown in \cite[Section 12.1--2]{CoutGod}.
The following result is directly implied by Lemma 12.3.1 in \cite{CoutGod}.

\begin{proposition}
\label{Prop:join}
Assume $X$ is the join of the $k_1$-regular graph $X_1$ on $m_1$ vertices and the  $k_2$-regular graph $X_2$ on $m_2$ vertices. Let
$$ \Delta =  k_1 + k_2  \pm \sqrt{ (k_1-k_2)^2 + 4 m_1 m_2 } . $$
Then $X$ admits $(V(X_1),V(X_1))$-GST and $(V(X_2),V(X_2))$-GST at time $\tau = 2\ell \pi/\sqrt{\Delta}$ for each integer $\ell$. $\Box$
\end{proposition}

\section{Examples}
\label{Sec:eg}

In previous sections we have seen mostly trivial examples of group state transfer, but also those cases that arise from perfect state transfer. We now discuss
non-trivial examples of this phenomenon.

\begin{theorem}
\label{Thm:bipartite}
Let $X$ be a connected bipartite  graph with bipartition $V(X) = V_0 \cup V_1$. 
\begin{itemize}
\item[(a)] If, for some  $\alpha >0$, all eigenvalues of $\alpha A$ are integers, then $X$ admits $(V_0,V_0)$-GST and $(V_1,V_1)$-GST at time 
$\tau  = \pi\alpha$;
\item[(b)] If, for some  $\alpha >0$, all eigenvalues of $\alpha A$ are odd integers, then $X$ admits $(V_0,V_1)$-GST and $(V_1,V_0)$-GST at time 
$\tau  = \pi\alpha/2$.
\end{itemize}
\end{theorem}

\begin{proof} 
Suppose $\theta_r$ is an eigenvalue of $X$ whose projector has block form 
$E_r = \left[ \begin{array}{c|c} F_{00} & F_{01} \\ \hline F_{10} & F_{11} \end{array} \right]$. Since $X$ is bipartite, there is an index $r'$ such that
$\theta_{r'} = - \theta_r$ and $E_{r'} = \left[ \begin{array}{r|r} F_{00} & -F_{01} \\ \hline -F_{10} & F_{11} \end{array} \right]$. So
$$ e^{i\theta_r \tau }    \left[ \begin{array}{c|c} F_{00} & F_{01} \\ \hline F_{10} & F_{11} \end{array} \right] 
  + e^{- i \theta_r \tau}  \left[ \begin{array}{r|r} F_{00} & -F_{01} \\ \hline -F_{10} & F_{11} \end{array} \right] =
  \left[ \begin{array}{r|r} ( e^{i\theta_r \tau } +  e^{- i \theta_r \tau} ) F_{00} &  ( e^{i\theta_r \tau } -  e^{- i \theta_r \tau} ) F_{01} \\ \hline  
  ( e^{i\theta_r \tau } -  e^{- i \theta_r \tau} )F_{10} &  ( e^{i\theta_r \tau } +  e^{- i \theta_r \tau} ) F_{11} \end{array} \right] .
$$
Let us assume first that $A$ is invertible so  that $A = \sum_{\theta_r > 0} \left( \theta_r E_r  + \theta_{r'} E_{r'} \right)$. Let us first consider case \textit{(b)}: at time
time $\tau= \pi \alpha/2$, $e^{i \theta_r \tau} = \pm i$, $e^{-i \theta_r \tau} = \mp i$  and
 the diagonal blocks of $e^{i\theta_r \tau } E_r + e^{i \theta_{r'} \tau} E_{r'}$  vanish.  Similarly, in case \textit{(a)}, the off-diagonal blocks of
 $e^{i\theta_r \tau } E_r + e^{i \theta_{r'} \tau} E_{r'}$  vanish at time $\tau = \pi\alpha/2$.  Summing over the positive eigenvalues $\theta_r$ gives our 
result, except in case \textit{(a)} where $A$ is singular.  To finish the argument we note that  the zero eigenspace of a bipartite graph admits 
a basis of eigenvectors each supported on just one of $V_0$, $V_1$. So the orthogonal projection $E_0$ is a block diagonal matrix and this does not affect 
the block diagonal structure of $U(\tau)$.
\end{proof}

\paragraph{The symmetric double star} In \cite{FanGod}, Fan and Godsil study pretty good state transfer on graphs composed of gluing together two stars. Let
$X$ be the graph (denoted $S_{k,k}$ in \cite{FanGod}) 
on vertex set $V(X)=\{1,\ldots,n\}$ where $n=2k+2$ with $E(X) = \{ 1a \mid 2\le a\le k+2\} \cup \{ 2a \mid k+3 \le a \le 2k+2\}$ 
and adjacency matrix
$$A= \left[ \begin{array}{c|c|c|c} 0 & 1  & \bj^\top & 0 \\  \hline  & & & \\ [-4mm]
1 & 0 & 0 &  \bj^\top  \\ \hline
\bj & 0 & 0 & 0 \\ \hline
0 & \bj & 0 & 0 \end{array} \right] $$
where $\bj$ is the $(m-1)\times 1$ matrix of all ones.  The eigenvalues are given, for example, in \cite{FanGod};
$$ \theta_0 = \frac12 \left( 1 + \sqrt{4k+1} \right), \ 
\theta_1 = \frac12 \left( -1 + \sqrt{4k+1}\right),  \   \theta_2 = 0,  \  
\theta_3 = \frac12 \left( 1 - \sqrt{4k+1}\right),  \  
\theta_4 = \frac12 \left( -1 - \sqrt{4k+1}\right) $$
with each nonzero eigenvalue having multiplicity one.

\begin{proposition}
\label{Prop:doublestar}
At time $\tau= 2\pi/ \sqrt{4k+1}$,  the symmetric double star $X$ admits $(S,S)$-GST for $S=\{1,2\}$.
\end{proposition}

\begin{proof}
Let $\sigma_0=\sigma_3=+1$  and $\sigma_1=\sigma_4=-1$ and note that $\theta_r^2-k= \sigma_r  \theta_r$ for $r\neq 2$.
The orthogonal projection onto the eigenspace of $A$ belonging to $\theta = \theta_r$ ($r\neq 2$) is
$$ E_r = \frac{1}{ 4k +2 \sigma_r \theta_r} \left[
\begin{array}{r|r|r|r}
  \theta_r^2  \phantom{\bj} &  \sigma_r  \theta_r^2  \phantom{\bj}& \theta_r \bj^\top  & \sigma_r \theta_r \bj^\top  \\ \hline  & & & \\[-3mm]
  \sigma_r \theta_r^2 \phantom{\bj} &   \theta_r^2 \phantom{\bj} & \sigma_r  \theta_r \bj^\top  & \theta_r \bj^\top  \\  \hline & & & \\[-3mm]
  \theta_r \bj  & \sigma_r  \theta_r \bj  &  J^{\phantom{\top}} & \sigma_r  J^{\phantom{\top}}  \\  \hline & & & \\[-3mm]
 \sigma_r     \theta_r \bj  & \theta_r \bj  &  \sigma_r J^{\phantom{\top}} &  J^{\phantom{\top}} \end{array} \right] $$
where $J$ is the $(m-1)\times (m-1)$ matrix of all ones.  The null space of $A$ is orthogonal to $e_1$ and $e_2$ so $E_2$ plays no role here. 
Since
$$\theta_3\tau = \theta_0\tau - 2\pi, \quad \theta_4\tau = \theta_1\tau - 2\pi,  \qquad   \theta_4=-\theta_0, \quad \theta_3=-\theta_1,$$
we have 
$$ e^{ i \theta_3 \tau} = e^{ i \theta_0 \tau }, \quad e^{ i \theta_4 \tau} =  e^{ i \theta_1 \tau } = \overline{ e^{i\theta_0 \tau} } ~ .$$
This gives us
$$ U(\tau) = e^{i \theta_0 \tau} (E_0 + E_3) + e^{i \theta_1 \tau} (E_1 + E_4) +  E_2   $$
having its first two rows equal to 
$$ \left[ \begin{array}{r|r|r|r}  U(\tau)_{11} & U(\tau)_{12} & \alpha  \bj^\top & \beta \bj^\top   \\  \hline & & & \\[-3mm]
 U(\tau)_{21} & U(\tau)_{22}  & \beta \bj^\top  & \alpha  \bj^\top \end{array} \right] $$
 where
 $$ \alpha = \sum_{r \neq 2} \frac{ e^{i \theta_r \tau}  \theta_r }{ 4k + 2 \sigma_r \theta_r}, \qquad 
 \beta = \sum_{r \neq 2} \frac{ e^{i \theta_r \tau} \sigma_r \theta_r }{ 4k + 2 \sigma_r \theta_r} ~ .$$
 Writing $K=4k+1$, we have 
 $$  \frac{ \theta_0 }{ 4k + 2  \theta_0} +\frac{ \theta_3 }{ 4k + 2  \theta_3} = \frac12 \left[  \frac{ 1 + \sqrt{K} }{K + \sqrt{K} } + \frac{ 1 - \sqrt{K} }{K - \sqrt{K} } \right] = 0$$
 and
 $$  \frac{ \theta_1 }{ 4k - 2  \theta_1} +\frac{ \theta_4 }{ 4k - 2  \theta_4} = \frac12 \left[  \frac{ -1 + \sqrt{K} }{K - \sqrt{K} } + \frac{ -1 - \sqrt{K} }{K + \sqrt{K} } \right] = 0$$ 
 giving us $\alpha=\beta=0$ as desired.
\end{proof}

Finally, we remark, without proof, that we also have $(S,S)$-GST for $S=\{3,6\}$ in  The McKay graph:
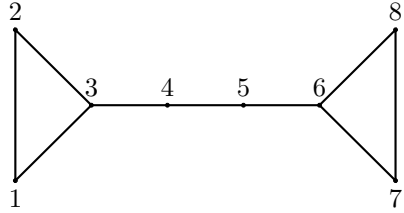
\begin{center}
\begin{tikzpicture}[scale=1]
\draw [fill] (0,-1) circle [radius=0.025]; \draw [fill] (0,1) circle [radius=0.025];
\draw [fill] (1,0) circle [radius=0.025]; 
\draw [fill] (2,0) circle [radius=0.025]; 
\draw [fill] (3,0) circle [radius=0.025]; 
\draw [fill] (4,0) circle [radius=0.025]; 
\draw [fill] (5,-1) circle [radius=0.025]; \draw [fill] (5,1) circle [radius=0.025];
\draw [thick] (1,0) -- (0,1) -- (0,-1) -- (1,0) -- (4,0) -- (5,1) -- (5,-1) -- (4,0);

\node [below] at (0,-1) {$1$}; \node [above] at (0,1) {$2$}; 
\node [above] at (1,0) {$3$};
\node [above] at (2,0) {$4$};
\node [above] at (3,0) {$5$}; 
\node [above] at (4,0) {$6$}; 
\node [below] at (5,-1) {$7$}; \node [above] at (5,1) {$8$}; 
\end{tikzpicture}  
\end{center}

\section{Some problems}

We now list some questions that we consider worthy of study.

\begin{enumerate}
\item Which graph products respect group state transfer? 
\item If $X$ is a path, can we classify all $S,T\subseteq V(X)$ for which bijective $(S,T)$-GST occurs?
\item Must every  $t$-closed subset $S$ of $V(X)$ arise from bijective group state transfer?
\item Which graph homomorphism are continuous with respect to the topologies of $t$-open sets?
\item Does case \textit{(d)} of Proposition \ref{Prop:srg} ever occur?
\item Suppose $X$ is a double star with $S$ and $T$ the natural partition of the vertices of degree one (elements of $S$ (resp., $T$) are pairwise at distance
two in $X$ and each $a\in S$ is at distance three from each $b\in T$. In what cases does $X$ admit $(S,T)$-GST?
\item Assume $X$ admits $(S,T)$-GST at time $\tau$ with $S\cap T=\emptyset$. When is there a weighted quotient graph $\bar{X}$  admitting PST from the (sole)
image of  $S$ to the sole image of $T$? (This is likely to be a rare occurrence.)
\item Is it true that, for almost all graphs $X$, the poset $\ST(X,t)$ is trivial for all $t\neq 0$?
\item Suppose $X$ admits $(S,T)$-GST at time $\tau$ and let $\delta$ denote the minimum distance from $a$ to $b$ over all $a\in S$, $b\in T$.  Must $|V(X)|$ grow exponentially with $\delta$?
\end{enumerate}

\section*{Acknowledgements}

We thank Ada Chan for useful discussions. Aleksandar Juri\u{s}i\'{c} helped with the proof of Proposition \ref{Prop:srg} and Chris Larsen
helped with the proof of Lemma \ref{Lem:usuallytrivial}.
The work of WJM was supported through a grant 
from the National Science Foundation (DMS Award \#1808376) which is gratefully acknowledged. 

\bibliography{wjmbibfile}

\end{document}